\documentclass[oribibl]{llncs}

\title{On the Approximability of Digraph Ordering}
\author{Sreyash Kenkre\inst{1}
	\and Vinayaka Pandit\inst{1}
        \and Manish Purohit\inst{2}\thanks{Partially supported by NSF grants CCF-1217890 and IIS-1451430.}
	\and Rishi Saket\inst{1}}

	\institute{IBM Research, Bangalore, Karnataka 560045, India \\ {\tt
		\{srekenkr, pvinayak, rissaket\}@in.ibm.com}  \and 
		University of Maryland, College Park, MD 20742, USA \\ {\tt
	manishp@cs.umd.edu}}

\usepackage{amssymb, amsmath, mathrsfs}
\usepackage{graphicx}
\usepackage[lined,boxed,commentsnumbered,linesnumbered]{algorithm2e}
\usepackage{algpseudocode}
\usepackage{color}
\usepackage{enumerate}
\usepackage{footmisc}
\usepackage{xspace}
\usepackage{url}

\DeclareMathOperator*{\argmax}{arg\,max}

\newcommand{\hide}[1]{{}}

\newcommand{\eps}{\varepsilon}

\newcommand{\mas}{{\sc MAS}\xspace}
\newcommand{\maxk}{{\sc Max-$k$-Ordering}\xspace}
\newcommand{\maxdi}{{\sc Max-DiCut}\xspace}
\newcommand{\mink}{{\sc DED$(k)$}\xspace}
\newcommand{\rmas}{{\sc RMAS}\xspace}
\newcommand{\rmasoff}{{\sc OffsetRMAS}\xspace}
\newcommand{\off}{{\theta}}

\newcommand{\slfrac}[2]{\left.#1\middle/#2\right.}
\newcommand{\mc}[1]{\ensuremath{\mathcal{#1}}\xspace}

\spnewtheorem{cl}{Claim}{\bfseries}{\itshape}

\usepackage{thmtools}
\usepackage{thm-restate}
\usepackage[margin=1.5in]{geometry}
\begin{document}

\maketitle

\begin{abstract}
	Given an $n$-vertex digraph $D = (V, A)$ the 
	\maxk problem is to compute a labeling $\ell : V \to [k]$ 
	maximizing the number of forward edges, i.e. edges 
	$(u,v)$ such that $\ell(u) < \ell(v)$. For
	different values of $k$, this reduces to
	\emph{maximum acyclic subgraph} ($k=n$),
	and \maxdi ($k=2$). This work studies the approximability
	of \maxk and its generalizations, motivated by their
	applications to job scheduling with \emph{soft} precedence
	constraints. 
	We give an LP rounding based
	$2$-approximation algorithm for \maxk for any $k=\{2,\dots, n\}$,
	improving on the known $\slfrac{2k}{(k-1)}$-approximation
	obtained via
	random assignment.
	The tightness of this
	rounding is shown by proving that for any $k=\{2,\dots, n\}$ and
	constant $\eps > 0$, 
	\maxk has an
	LP integrality gap of $2 - \eps$ for
	$n^{\Omega\left(\slfrac{1}{\log\log k}\right)}$ rounds of
	the Sherali-Adams hierarchy.

	\smallskip
	A further generalization of \maxk is the \emph{restricted
	maximum acyclic subgraph} problem or \rmas, where each vertex
	$v$ has a finite set of allowable labels $S_v \subseteq
	\mathbb{Z}^+$. We prove an LP rounding based
	$\slfrac{4\sqrt{2}}{\left(\sqrt{2}+1\right)} \approx 2.344$
	approximation for it, improving on the $2\sqrt{2} \approx
	2.828$ approximation recently given by Grandoni et
	al.~\cite{grandoni2015lp}. In fact, our approximation algorithm also works for a
	general version where the objective counts the edges which go
	forward by at least a positive \emph{offset} specific to each
	edge. 
	
	\smallskip
	The minimization formulation of digraph
	ordering is \emph{DAG edge deletion} or \mink, which
	requires deleting the minimum number of edges from an
	$n$-vertex directed
	acyclic graph (DAG) to remove all paths of length $k$.
	We show that both, the LP relaxation
	and a local ratio approach for \mink yield $k$-approximation 
	for any
	$k\in [n]$.
	A vertex deletion
	version was studied earlier by Paik et al.~\cite{paik1994deleting}, 
	and
	Svensson~\cite{svensson2012hardness}. 
\end{abstract}

\section{Introduction}
\label{sec:introduction}
One of the most well studied combinatorial problems on directed graphs
(digraphs)
is the \emph{Maximum Acyclic Subgraph} problem (\mas): given an $n$-vertex digraph,
find a subgraph\footnote{Unless specified, throughout this
paper a \emph{subgraph} is not necessarily induced.} of maximum number 
of edges 
containing no directed cycles. 
 An equivalent formulation of \mas is to
obtain a linear ordering of the vertices which maximizes the number of directed
edges going forward. A natural generalization is \maxk
where the goal is to compute the best \emph{$k$-ordering}, i.e. a labeling of
the vertices from $[k] = \{1,\dots, k\}$ ($2\leq k \leq n$), which
maximizes the number of directed edges going forward in this ordering.
It can be seen -- and we show this formally -- that \maxk is
equivalent to finding the maximum subgraph which has no directed
cycles, and no directed paths\footnote{The length of a directed path is
the number of directed edges it contains.} of length $k$. Note that \mas is
the special case of \maxk when $k = n$, and for $k=2$ \maxk reduces to the
well known \maxdi problem. 

A related problem is the \emph{Restricted Maximum Acyclic Subgraph} problem or
\rmas, 
in which each vertex $v$ of the digraph has
to be assigned a label from a finite set 
$S_v\subseteq \mathbb{Z}^+$ to maximize the number of edges
going forward in this assignment. 
Khandekar et al.~\cite{khandekar2009hardness} 
introduced \rmas in the context of \emph{graph
pricing} problems and its approximability has recently 
been studied by Grandoni et al.~\cite{grandoni2015lp}. A further
generalization is \rmasoff where each edge $(u,v)$ has an offset
$o_e \in \mathbb{Z}^+$ and is satisfied by a labeling $\ell$ if
$\ell(u) + o_e \leq \ell(v)$. Note that
when all offsets are unit \rmasoff reduces to \rmas, which in turn
reduces to \maxk when $S_v = [k]$ for all vertices $v$.

This study focuses on the approximability of \maxk and its
generalizations  and 
is motivated by their applicability
in scheduling jobs with \emph{soft} precedences under a hard deadline. 
Consider the following simple case
of discrete time scheduling: given $n$ unit length jobs with precedence
constraints and an infinite capacity machine, find a schedule so that all the jobs
are completed by timestep $k$. Since it may not be feasible to satisfy all
the precedence constraints, the goal is to satisfy the maximum number.
This is equivalent to \maxk on the corresponding precedence digraph.
One can generalize this setting to each job having a
set of allowable timesteps when it can be
scheduled. This can be abstracted as \rmas and a further
generalization to each
precedence having an associated lag between the start-times yields
\rmasoff as the underlying optimization problem.  

Also of interest is the minimization version of \maxk on directed
\emph{acyclic} graphs (DAGs). We refer to it as \emph{DAG edge
deletion} or \mink where the goal is to
delete the minimum number of directed edges from a DAG so that the
remaining digraph does not contain any path of length $k$. Note that
the problem for arbitrary $k$ does not admit any approximation factor on
general digraphs since even detecting whether a digraph has a 
path of length $k$ is the well studied NP-hard longest path problem. …
A vertex deletion formulation of
\mink was introduced as an abstraction of
certain VLSI design and communication problems by Paik et
al.~\cite{paik1994deleting} who gave efficient algorithms for it on special
cases of DAGs, and proved it to be NP-complete in general. More recently,
its connection to project scheduling was noted 
by Svensson~\cite{svensson2012hardness} who proved 
inapproximability results for
the vertex deletion version. 
  
The rest of this section gives a background of previous related work,
describes our results, and provides an overview of the techniques
used.

\subsection{Related Work}
It is easy to see that \mas admits a trivial $2$-approximation, by
taking any linear ordering or its reverse, and this is also obtained by a random
ordering.
For \maxk the random $k$-ordering yields a
$\slfrac{2k}{(k-1)}$-approximation for any $k \in \{2,\dots, n\}$. For
$k=2$, which is \maxdi, the semidefinite programming (SDP) 
relaxation is shown to yield a $\approx
1.144$-approximation in \cite{lewin2002improved}, improving upon previous analyses of
\cite{matuura20010},  \cite{zwick2000analyzing}, and \cite{FeigeGoemans}. 
As mentioned above, \rmas is a generalization
of \maxk, and a $2\sqrt{2}$-approximation for it
based on linear programming (LP) 
rounding was shown recently by Grandoni et al.~\cite{grandoni2015lp}
which is also the best approximation for \maxk for $k=3$. 
For $4\leq k \leq n-1$, to the best of our knowledge the proven
approximation factor for \maxk remains $\slfrac{2k}{(k-1)}$. 

On the hardness side, Newman~\cite{newman2000approximating} showed that 
\mas is NP-hard to
approximate within a factor of $\slfrac{66}{65}$. Assuming
Khot's~\cite{khot2002power} Unique
Games Conjecture (UGC), Guruswami et al.~\cite{guruswami2008beating} gave a
$(2-\eps)$-inapproximability for any $\eps > 0$.
Note that \maxdi is at least as hard as {\sc Max-Cut}. Thus, 
for $k=2$, \maxk is NP-hard to
approximate within factor $(\slfrac{13}{12}
-\eps)$~\cite{hastad2001some},
and within factor  $1.1382$ assuming the UGC~\cite{khot2007optimal}. 
For larger constants $k$, the
result of Guruswami et al.~\cite{guruswami2008beating} implicitly shows a
$(2-o_k(1))$-inapproximability for \maxk, assuming the UGC.

For the vertex deletion version of 
\mink, Paik et al.~\cite{paik1994deleting} gave linear time and quadratic time 
algorithms for rooted trees and series-parallel
graphs respectively. 
The problem reduces to vertex cover on $k$-uniform hypergraphs
for any constant $k$ thereby admitting a $k$-approximation, and a
matching $(k-\eps)$-inapproximability assuming the UGC was obtained by
Svensson~\cite{svensson2012hardness}.

\subsection{Our Results}
The main algorithmic result of this paper is the following improved
approximation guarantee for \maxk.
\begin{theorem}\label{thm:main1}
	There exists a polynomial time $2$-approximation algorithm for \maxk on
	$n$-vertex weighted digraphs for any $k \in \{2,\dots, n\}$.
\end{theorem}
The above approximation is obtained by appropriately rounding the
standard LP relaxation of the CSP formulation of \maxk. For small values of $k$
this yields significant improvement on
the previously known approximation factors: $2\sqrt{2}$ for 
$k=3$ (implicit in \cite{grandoni2015lp}), $\slfrac{8}{3}$ for $k=4$, 
and $2.5$ for $k=5$. The latter two factors follow from the previous best 
$\slfrac{2k}{(k-1)}$-approximation given by a random $k$-ordering
for $4\leq k \leq n-1$. The detailed proof of Theorem \ref{thm:main1} 
is given in
Section \ref{sec:lp-rounding}.

Using an LP rounding approach similar to Theorem \ref{thm:main1}, in
Section \ref{sec:generalizations} we
show the following improved approximation for \rmasoff which implies
the same for \rmas. Our result improves the 
previous $2\sqrt{2}\approx 2.828$-approximation for \rmas 
obtained by Grandoni et al.~\cite{grandoni2015lp}.
\begin{theorem}\label{thm:rmasoff}
	There exists a polynomial time
	$\slfrac{4\sqrt{2}}{(\sqrt{2}+1)} \approx 2.344$ approximation algorithm 
	for \rmasoff on weighted digraphs. 
\end{theorem}

Our next result gives a lower bound that matches the approximation
obtained in Theorem \ref{thm:main1}. In Section \ref{sec:integrality-gaps}, we show that even after
strengthening the LP relaxation of \maxk with a large number of rounds of the
Sherali-Adams hierarchy, its integrality gap remains close to $2$, and
hence Theorem \ref{thm:main1} is tight. 
\begin{theorem}\label{thm:main2}
	For any small enough constant $\eps > 0$, there exists $\gamma
	> 0$ such that for \maxk on $n$-vertex weighted digraphs and
	any $k \in \{2, \ldots, n\}$, 
	the LP relaxation with $n^{\left(\slfrac{\gamma}{\log\log
	k}\right)}$ 
	rounds of
	Sherali-Adams constraints has a $(2-\eps)$ integrality gap.
\end{theorem}

For \mink on DAGs we prove in Section \ref{sec:mink} 
the following approximation for any $k$,
not necessarily a constant.
\begin{theorem}\label{thm:deletion1}
	The standard LP relaxation for \mink on $n$-vertex DAGs can be
	solved in polynomial time for $k=\{2,\dots, n-1\}$ and yields
	a $k$-approximation. The same approximation factor is also
	obtained by a combinatorial algorithm.
\end{theorem}
We complement the above by showing in Section \ref{sec:mink-hardness} a
$\left(\lfloor\slfrac{k}{2}\rfloor-\eps\right)$ hardness factor for \mink via a
 simple gadget reduction from
Svensson's~\cite{svensson2012hardness} $(k-\eps)$-inapproximability 
for the vertex deletion version for constant $k$, assuming the
UGC.

\subsection{Overview of Techniques}

The approximation algorithms we obtain for \maxk and its
generalizations are based on rounding the standard LP relaxation for
the instance. \maxk is viewed as a constraint satisfaction problem (CSP) 
over alphabet
$[k]$, and the corresponding LP relaxation has $[0,1]$-valued variables 
$x^v_i$ for each vertex $v$ and label $i \in [k]$, and $y^e_{ij}$ for
each edge $(u,v)$ and pairs of labels $i$ and $j$ to $u$ and $v$
respectively. We show that a generalization of the 
rounding algorithm used by Trevisan~\cite{trevisan1998parallel} for approximating 
$q$-ary boolean CSPs yields a
$2$-approximation in our setting. The key
ingredient in the analysis is a lower bound
on a certain product of the $\{x^u_i\}, \{x^v_i\}$ variables 
corresponding to the end points of an edge $e=(u,v)$
in terms of the $\{y^e_{ij}\}$ variables for that edge. This improves 
a weaker bound shown by Grandoni et al.~\cite{grandoni2015lp}.
For \rmasoff, a
modification of this rounding algorithm yields the improved
approximation. 

The construction of the integrality gap for the LP
augmented with  Sherali-Adams constraints for \maxk begins
with a simple integrality gap instance for the basic LP
relaxation. This instance is appropriately sparsified  to ensure
that subgraphs of polynomially 
large (but bounded) size are \emph{tree-like}.
On trees, it is easy to construct a distribution over  
labelings from $[k]$
to the vertices (thought of as $k$-orderings), 
such that the marginal distribution on each vertex is uniform
over $[k]$ and a large fraction of edges are satisfied in expectation.
Using this along with the sparsification allows us to construct 
distributions for each bounded subgraph, i.e. good local distributions.
Finally a geometric 
\emph{embedding} of the marginals of these distributions
followed by Gaussian 
rounding yields modified local distributions which are
\emph{consistent} on the common vertex sets. These distributions
correspond to an LP solution with a high objective value,  
for large number of rounds of Sherali-Adams constraints. Our
construction follows the approach in a recent work of
Lee~\cite{lee2014hardness} which is based on earlier works of Arora et
al.~\cite{arora2002proving} and Charikar et al.~\cite{charikar2009integrality}.

For the \mink problem, the approximation algorithms stated in Theorem
\ref{thm:deletion1} are obtained using the acyclicity of the input
DAG. In particular, we show that both, the LP rounding and the
local ratio approach, can be implemented in polynomial time on DAGs 
yielding  $k$-approximate solutions.

\section{Preliminaries} \label{sec:prelims}
This section formally defines the problems studied in this paper. We
begin with \maxk.
\begin{definition}\label{def:maxk} \maxk: Given an $n$-vertex digraph $D =
	(V,A)$ with a non-negative weight function $w : A \to
	\mathbb{R}^+$, 
	and an integer $2 \leq k \leq n$, find a labeling to the
	vertices $\ell: V \rightarrow [k]$ that maximizes the weighted 
	fraction of edges $e = (u,v) \in A$ such that $\ell(u) <
	\ell(v)$, i.e. forward edges.
\end{definition}
It can be seen that \maxk is equivalent to the problem of computing
the maximum weighted subgraph of $D$ which is acyclic and does not
contain any directed path of length $k$. The
following lemma implies this equivalence.
\begin{restatable}{lemma}{equivalence}\label{lem:equiv}
	Given a digraph $D = (V,A)$, there exists a labeling
	$\ell : V \to [k]$ with each edge $e = (u,v) \in A$ satisfying
	$\ell(u)< \ell(v)$, if and only if $D$ is acyclic and does not
	contain any directed path of length $k$.
\end{restatable}

\begin{proof}
	If such a labeling $\ell$ exists then every edge is directed
	from a lower labeled vertex to a higher labeled one. Thus, 
	there are no
	directed cycles in $D$. Furthermore, any directed path in $D$ has
	at most $k$ vertices on it, and is of length at most $k-1$. On
	the other hand, if $D$ satisfies the second condition in the
	lemma, then
	choose $\ell(v)$ for any vertex $v$ to be $t_v + 1$, where
	$t_v$ is the length of the
	longest path from any source to $v$. It is easy to see that
	$\ell(v) \in [k]$ and for each edge $(u,v)$, $\ell(u) <
	\ell(v)$. \qed
\end{proof}

The generalizations of \maxk studied in this work, viz. 
\rmas and \rmasoff, are defined as
follows.
\begin{definition}\label{def:rmasoff}\rmasoff: 
	The input is a digraph $D = (V, A)$ with a
	finite subset $S_v \subseteq \mathbb{Z}^+$ of labels for each
	vertex $v \in V$, a non-negative weight function $w :
	A\to\mathbb{R}^+$, and offsets $o_e \in \mathbb{Z}^+$ for each
	edge $e \in A$. A labeling $\ell$ to $V$ s.t. $\ell(v) \in S_v,
	\forall v \in V$ satisfies an edge $e= (u,v)$ if $\ell(u) + o_e \leq
	\ell(v)$. The goal is to compute a labeling that maximizes the
	weighted fraction of satisfied edges. \rmas is the special
	case when each offset is unit.
\end{definition}

As mentioned earlier, \mink is not approximable on general digraphs. 
Therefore, we define it only on DAGs.
\begin{definition}\mink:  Given a DAG $D =
	(V,A)$ with a non-negative weight function $w : A \to
	\mathbb{R}^+$, 
	and an integer $2 \leq k \leq n-1$, find a minimum weight set
	of edges $F \subseteq A$ such that $(V,A\setminus F)$ does not
	contain any path of length $k$.
\end{definition}
The rest of this section describes the LP relaxations for \maxk and
\rmasoff studied in this paper. 

\subsection{LP Relaxation for \maxk}
From Definition \ref{def:maxk}, an instance $\mc{I}$ of 
\maxk is given by $D = (V, A)$, $k$, and $w$. 
Viewing it as a CSP over label set $[k]$, the standard LP
relaxation given in Figure \ref{fig:LPmaxk} is defined over 
variables $x^v_i$ for each vertex $v$ and
label $i$, and $y^{e}_{ij}$ for each edge $e = (u,v)$ and labels $i$
to $u$ and $j$ to $v$.
\begin{figure}[h]
\setlength{\fboxsep}{10pt}
\begin{center}
\fbox{
\begin{minipage}[c]{4.3in}
\begin{equation}
	\max
	\sum_{e \in A}w(e)\cdot\sum_{\substack{i,j \in [k]\\i <
	j}}y^e_{ij}
\end{equation}
subject to,
\begin{align}
	& \forall v \in V, && \displaystyle \sum_{i \in [k]} x^v_i
	= 1 .
\label{eqn:LPvsum}\\
\ &\ & \nonumber\\
  & \forall e = (u,v) \in A, \textnormal{ and } i,j \in [k], &&
\displaystyle \sum_{\ell \in [k]}y^e_{i\ell} = x^u_{i}, 
\label{eqn:LPuedgesum} \\
&\textnormal{and,}&&\displaystyle \sum_{\ell \in [k]}y^e_{\ell j} =
x^v_{j}. 
\label{eqn:LPvedgesum} \\
&&& \nonumber \\
&\forall v \in V,\textnormal{ and } i \in [k], && x^v_i \geq 0.
	\label{eqn:LPxpositive}\\
&\forall e \in A,\textnormal{ and } i,j \in [k],  && y^e_{ij} \geq 0.
\label{eqn:LPypositive}
\end{align}
\caption{LP Relaxation for instance $\mc{I}$ of \maxk.}
\label{fig:LPmaxk}
\end{minipage}
}
\end{center}
\end{figure}

\noindent
{\bf Sherali-Adams Constraints.} Let $z^S_\sigma \in [0,1]$ 
be a variable corresponding
to a subset $S$ of vertices, and a labeling 
$\sigma : S\to [k]$. The LP relaxation in Figure
\ref{fig:LPmaxk} can augmented with $r$ rounds of Sherali-Adams
constraints which are 
defined over the variables $\{z^S_\sigma\ \mid\ 1 \leq |S| \leq
r+1\}$. The additional constraints are given in Figure \ref{fig:SAmaxk}. The
Sherali-Adams variables define, for each subset $S$ of at most
$(r+1)$ vertices, a distribution over the possible labelings from
$[k]$ to the vertices in $S$. The constraints given by Equation 
\eqref{eqn:SAcons}  ensure that these
distributions are consistent across subsets. Additionally, 
Equations \eqref{eqn:SAxcons} and \eqref{eqn:SAycons} ensure the
consistency of these distributions with the variables of the
standard LP relaxation given in Figure \ref{fig:LPmaxk}.
\begin{figure}[h]
\setlength{\fboxsep}{10pt}
\begin{center}
\fbox{
\begin{minipage}[c]{4.3in}
\begin{align}
	&\forall S\subseteq T\subseteq V, && \nonumber \\
	& 1\leq |S|, |T|\leq r+1, && \nonumber \\	
	& \textnormal{and }
	\sigma:S\to [k], &&  z^S_\sigma = \sum_{\substack{\rho: T\to [k]\\
	\rho|_S = \sigma}} z^T_\rho. \label{eqn:SAcons} \\
	& \forall S\subseteq V, 1\leq |S|\leq r+1, && \nonumber \\
	& \textnormal{and }\sigma:S\to [k],&& 0 \leq 
	z^S_\sigma \leq 1. \\
	&&& \nonumber \\
	&\forall v\in V,\textnormal{ and }\sigma:\{v\}\to [k], &&
	\nonumber \\
	&\textnormal{s.t. }\sigma(v)=i, && x^v_i = 
	z^{\{v\}}_\sigma. \label{eqn:SAxcons} \\
	&&&\nonumber \\
	&\forall e=(u,v)\in A,\textnormal{ and,} && \nonumber \\ 
	& \sigma:\{u,v\}\to
	[k], &&  \nonumber \\ 
	&\textnormal{s.t. }(\sigma(u),\sigma(v)) = (i,j), && 
	y^e_{ij} = z^{\{u,v\}}_\sigma. \label{eqn:SAycons}
\end{align}
\caption{$r$-round Sherali-Adams constraints for LP relaxation in
Figure \ref{fig:LPmaxk}.}
\label{fig:SAmaxk}
\end{minipage}
}\end{center}
\end{figure}

\smallskip
\noindent
{\bf LP Relaxation for \rmas and \rmasoff.} 
The LP relaxation for \rmas is
a generalization of the one in Figure
\ref{fig:LPmaxk} for \maxk and we omit a detailed definition. Let $\mathcal{S} = \cup_{v \in V} S_v$ denote the set of all labels.
For convenience, we define variables $\{x^v_i\,\mid\, v\in V, i\in \mathcal{S}\}$ and $\{y^e_{ij}\,\mid\,
e=(u,v)\in A, i,j\in \mathcal{S}\}$ and force the infeasible assignments to be zero, i.e. $x^v_i = 0$ for $i \notin S_v$. The other constraints
are modified accordingly.  For \rmasoff, an additional change is
that the contribution to the objective from each edge $e = (u,v)$ is
$ \sum_{i\in S_u, j\in S_v, i + o_e\leq j}y^e_{ij}$. 

\section{A 2-Approximation for \maxk}
\label{sec:lp-rounding}

This section proves the following theorem that implies 
Theorem \ref{thm:main1}.
\begin{theorem}
\label{thm:2-approx} Let $\{x^v_i\}, \{y^e_{ij}\}$ denote an optimal solution
to the LP in Figure \ref{fig:LPmaxk}. Let $\ell : V\to[k]$ 
be a randomized labeling
obtained by independently assigning to each vertex $v$ 
label $i$  with probability
$\slfrac{1}{2k} + \slfrac{x^v_i}{2}$. Then, for any edge $e=(u,v)$,
\[\Pr[ \ell(u) < \ell(v)] \geq \frac{1}{2} \left(\sum_{\substack{i,j \in [k]\\i <
j}}y^e_{ij}\right).\] 
\end{theorem}
To analyze the rounding given above, we  need the following key 
lemma that bounds the sum of
products of row and column sums of a matrix in terms of the matrix
entries. It improves a weaker bound shown by Grandoni
et al.~\cite{grandoni2015lp} and also generalizes to arbitrary
offsets.
\begin{lemma}
\label{lem:psd}
Let $\mathbb{A} = [a_{ij}]$ be a $k \times k$ matrix with non-negative entries.
Let $r_i = \sum_j a_{ij}$ and $c_j = \sum_{i} a_{ij}$ denote the sum
of entries in the $i^\text{th}$ row and $j^\text{th}$ column
respectively, and let $1 \leq \off \leq k-1$ be an integer offset . 
Then,
\begin{equation}
\sum_{\substack{i + \off \leq j \\ i,j \in [k]}} r_i c_j \geq
\frac{k-\off+1}{2(k-\off)} \left(\sum_{\substack{i + \off \leq j \\ i,j \in [k]}} 
a_{ij}\right)^2. \label{eqn:lempsd}
\end{equation}
\end{lemma}
\begin{proof} The LHS of the above is simplified as,
\begin{align} 
	\sum_{i + \off \leq j}r_ic_j &= \sum_{i+\off \leq j}
	\left[\bigg(\sum_{j'} a_{ij'}\bigg)\bigg(\sum_{i'}
	a_{i'j}\bigg)\right] \label{eqn:RHS1} \\ 
	&\geq \sum_{x + \off \leq y} a_{xy}^2 +
	2\cdot\sum_{\substack{x+\off \leq y\\x 
	+ \off \leq y'\\ y < y'}} a_{xy}a_{xy'} 
	+ \sum_{\substack{x+\off \leq
	y \\x'+\off \leq y'\\x < x'}} 
	a_{xy}a_{x'y'}, \label{eqn:RHS2}
\end{align}
where all the indices above are in $[k]$. Note that \eqref{eqn:RHS2}
follows from \eqref{eqn:RHS1} because: \newline
      (i) For any $x + \off \leq y$, $a_{xy}^2$ appears in
      the RHS of \eqref{eqn:RHS1} when $i=x$ and $j=y$.\newline
      (ii) For $x + \off \leq y$ and $x + \off \leq y'$,
      $a_{xy}a_{xy'}$ appears in the RHS of \eqref{eqn:RHS1} both,
      when $i=x, j=y$, and when $i=x, j=y'$.\newline
      (iii) For any $x + \off \leq y$ and $x' + \off \leq y'$ (say
      $x<x'$), it must be that $x + \off \leq y'$, and hence
      $a_{xy}a_{x'y'}$ appears in the RHS of \eqref{eqn:RHS1}
      when $i=x$ and $j=y'$.\newline
      Thus, we obtain,
      \begin{equation}
	\sum_{i + \off \leq j}r_ic_j \geq 
     \left(\sum_{x + \off \leq y} a_{xy}\right)^2 - 
     \left(\sum_{\substack{x + \off \leq y\\x' + \off \leq y'\\x<x'}}
	     a_{xy}a_{x'y'}\right).
     \end{equation}
Therefore, it is sufficient to show that
\begin{align}
 &\sum_{\substack{x + \off \leq y\\x' + \off \leq y'\\x<x'}}
	a_{xy}a_{x'y'} \leq \frac{k-\off-1}{2(k-\off)}\left(\sum_{x +
	\off \leq y} a_{xy}\right)^2.\label{eqn:squared1}
\end{align}
Substituting, 
$$\left(\sum_{x +\off \leq y} a_{xy}\right)^2 = 
\sum_{x + \off \leq y} a_{xy}^2 + 2\cdot\sum_{\substack{x + \off \leq y\\x
+ \off \leq y'\\ y < y'}} a_{xy}a_{xy'} + 2\cdot\sum_{\substack{x +
\off \leq y\\x' + \off \leq y'\\x<x'}} a_{xy}a_{x'y'},$$
and simplifying, inequality \eqref{eqn:squared1} can be rewritten as,
\begin{align}
 \sum_{x + \off \leq y} a_{xy}^2 + 2\sum_{\substack{x + \off \leq y\\x
+ \off \leq y'\\y  < y'}}&a_{xy}a_{xy'} - \left(\frac{2}{k-\off -1}\right)\cdot
\sum_{\substack{x
+ \off \leq y\\x' + \off \leq y'\\x<x'}} a_{xy}a_{x'y'}\,\geq\, 0 ,
\label{eq:toprove} \\
\Leftrightarrow &\ \ \overline{a}^{\sf T}M\overline{a}\,\geq\, 0, 
  \end{align}
  where $\overline{a} \in \mathbb{R}^{\mathcal{Z}}$,
  $\mathcal{Z} := \{(x,y)\,\mid\, x + \off \leq y \text{ and } x,y \in [k]\}$ with 
  $\overline{a}_{(x,y)} :=
  a_{xy}$, and $M \in \mathbb{R}^{\mathcal{Z}\times\mathcal{Z}}$ 
  is a symmetric matrix defined as follows: 
\begin{align}
	M_{(x,y)(x',y')} = 
  \begin{cases}
	  1 &\text{if $(x,y) = (x',y')$,}\\
    1 &\text{if $x'=x$, and $y\neq y'$,}\\
	  \slfrac{-1}{(k-\off-1)} &\text{if $x\neq x'$.}
  \end{cases}
\end{align}
To complete the proof of the lemma we show that $M$ is positive
semidefinite. Consider the set of unit vectors
$\{v_x\,\mid\, 1\leq x\leq k-\off\}$ given by the normalized corner
points of the $(k-\off-1)$-dimensional simplex centered at the origin. 
It is easy to see
(for e.g. in Lemma 3 of \cite{frieze1997improved}) that, $\langle
v_x,v_{x'}\rangle = \slfrac{-1}{(k-\off-1)}$ if $x \neq x'$. 
Thus, $M = L^{\sf T}L$, where $L$ is a matrix whose columns are
indexed by $\mathcal{Z}$ such that 
the $(x,y)$ column is $v_x$. 
Therefore, $M$ is positive 
semidefinite.
\qed
\end{proof}
\begin{proof}[of Theorem \ref{thm:2-approx}]
For brevity, let $z_e = \sum_{i <j}y^e_{ij}$ denote the contribution
of the edge $e$ to the LP objective. From the definition of the
rounding procedure we have, 
\begin{align}
\Pr[\ell(u) < \ell(v)] &= \sum_{i<j} \Pr[\ell(u) = i] \Pr[\ell(v)=j] \nonumber \\
&= \sum_{i<j} \left(\frac{1}{2k} + \frac{x^u_{i}}{2}\right)\left(\frac{1}{2k} + 
	\frac{x^v_{j}}{2}\right)\nonumber \\
&= \frac{1}{4}\left(\frac{(k-1)}{2k} + \frac{1}{k} \sum_{i<j}(x^u_{i}
+ x^v_{j}) + \sum_{i<j} x^u_{i}x^v_{j} \right)\nonumber 
\intertext{We can now apply Lemma \ref{lem:psd} to the $k \times k$ matrix $[y^e_{ij}]$.  The LP constraints guarantee that $r_i = x^u_i$ and $c_j = x^v_j$ are equal to the row and column sums respectively. Further, substituting offset $\off = 1$, we obtain }\
\Pr[\ell(u) < \ell(v)] &\geq \frac{1}{4}\left(\frac{(k-1)}{2k} + \frac{1}{k} \sum_{i<j}(x^u_{i} + x^v_{j}) +  \frac{k}{2(k-1)}z_e^2 \right). \label{eq:A}\\
\intertext{On the other hand,}
\sum_{i<j}(x^u_{i} + x^v_{j}) &= \sum_{i=1}^{k-1} (k-i)x^u_{i} +
\sum_{j=2}^{k} (j-1)x^v_{j} \nonumber \\
&\geq \sum_{i=1}^{k-1} \left[(k-i) \sum_{j'>i} y^e_{ij'} \right] +
\sum_{j=2}^{k} \left[(j-1)\sum_{i'<j} y^e_{i'j} \right].
\label{eqn:expanded}
\end{align}
For $a<b$, $y^e_{ab}$ appears $(k-a)$ times in the RHS of the above
inequality when $i=a$, and $(b-1)$ times when $j=b$. Since $k-a+b-1
\geq k$, we obtain that RHS of Equation \eqref{eqn:expanded} is 
lower bounded by $k\sum_{a<b}y^e_{ab} = kz_e$. Substituting back into Equation
\eqref{eq:A} and simplifying gives us that $\Pr[\ell(u) < \ell(v)]$ is at least,
\begin{equation}
	\frac{z_e}{4}\left[1 + \frac{1}{2}\left(\frac{(k-1)}{kz_e} +  
	\frac{kz_e}{(k-1)}\right)\right] \geq \frac{z_e}{4}\left(1 +
	1\right) = \frac{z_e}{2},
\end{equation}
where we use $t + 1/t \geq 2$ for $t > 0$.
\qed
\end{proof}

\section{Approximation for \rmasoff}
\label{sec:generalizations}
Let $D = (V,A)$, $\{S_v\}_{v\in V}$, $w$, and $\{o_e\}_{e\in A}$ 
constitute an instance of \rmasoff
as given in Definition \ref{def:rmasoff}.
Without loss of generality, one can assume that for each edge $e =
(u,v) \in A$,  $\min(S_u) + o_e \leq \max(S_v)$, otherwise
no feasible solution can satisfy $e$ and that edge can be removed. 
A simple randomized
strategy that independently assigns each vertex $v$ either $\ell^v_{min}
:= \min(S_v)$
or $\ell^v_{max} := \max(S_v)$
with equal probability is a $4$-approximation. The
recent work of Grandoni et al.~\cite{grandoni2015lp} show that
combining this randomized scheme with an appropriate LP-rounding
yields a $2\sqrt{2} \approx 2.828$ approximation algorithm for
\rmas.

We show that a variant of the rounding scheme
developed in Section \ref{sec:lp-rounding} yields an improved
approximation factor for \rmasoff. 
In particular, we prove the following theorem which
implies Theorem \ref{thm:rmasoff}.
\begin{theorem}
\label{thm:rmasapprox}
Let $\{x^v_i\}, \{y^e_{ij}\}$ 
denote an optimal solution to the linear programming relaxation of
\rmasoff described in Section \ref{sec:prelims}. Let $\ell$ be a
randomized labeling obtained by independently assigning labels to
each vertex $v$ with the following probabilities:
  \begin{align}
\Pr[ \ell(v) = i] = 
\begin{cases}
\frac{1}{4} + \frac{x^v_i}{2} \quad &\text{if } i \in \{\ell^v_{min}, \ell^v_{max}\}\\
\frac{x^v_i}{2} &\text{if } i \in S_v \setminus \{\ell^v_{min}, \ell^v_{max}\}
\end{cases}\label{eqn:rmasoffprobs}   
\end{align}
Then, for any edge $e=(u,v)$ we have
\[\Pr[ \ell(u)+o_e \leq \ell(v)] \geq \frac{1}{4}\left(1 + \frac{1}{\sqrt{2}} \right) \left(\sum_{\substack{i\in S_u, j\in S_v\\ i+o_e \leq j}}y^e_{ij}\right).\]
\end{theorem}
\begin{proof}
Let $\mathcal{S} = \cup_{v \in V} S_v$ denote the set of all labels
and let $z_e = \left(\sum_{\substack{i+o_e \leq j}}y^e_{ij}\right)$
denote the contribution of the edge $e$ to the LP objective. We have,
  \begin{align*}
\Pr[ \ell(u) + o_e \leq \ell(v)] &= \sum_{\substack{i+o_e \leq j\\i\in
S_u, j \in S_v}} \Pr[\ell(u)=i] \Pr[\ell(v)=j]
  \end{align*}
  Substituting the assignment probabilities from \eqref{eqn:rmasoffprobs}
  into the above and simplifying we obtain,
  \begin{eqnarray}
	  &&\Pr[ \ell(u) + o_e \leq \ell(v)]\nonumber \\ &=&
  \frac{1}{16} + \frac{1}{8}\left(\sum_{\substack{i \leq \ell^v_{max} - o_e \\ i \in \mathcal{S}}} x^u_{i} + \sum_{\substack{j \geq \ell^u_{min} + o_e \\ j \in \mathcal{S}}} x^v_{j}\right) + \frac{1}{4}\left(\sum_{\substack{i+o_e \leq j \\ i,j \in \mathcal{S}}} x^u_{i}x^v_{j}\right)
  \end{eqnarray}
Note that we allow $i,j \in \mathcal{S}$ in the above sums instead of
$S_u$ and $S_v$. This does not affect the analysis as the LP forces $x^u_i
= 0$ for $i \notin S_u$ and similarly for $v$.  Now, consider the
$|\mathcal{S}| \times |\mathcal{S}|$ matrix $[y^e_{ij}]$. Since
$x^u_i$ and $x^v_j$ are equal to the row sums and column sums of this
matrix respectively, Lemma \ref{lem:psd} guarantees that,
\begin{equation}
\sum_{\substack{i+o_e \leq j \\ i,j \in \mathcal{S}}} x^u_{i}x^v_{j}
\geq \frac{|\mathcal{S}| - o_e + 1}{2(|\mathcal{S}| - o_e)}
\left(\sum_{\substack{i\in S_u, j\in S_v\\ i+o_e \leq j}}y^e_{ij}
\right)^2 \geq \frac{(|\mathcal{S}| - o_e + 1)}{2(|\mathcal{S}| -
o_e)}z_e^2 \geq \frac{z_e^2}{2}. \nonumber
\end{equation}
We thus have, 
\begin{eqnarray}
	&& \Pr[ \ell(u)+o_e \leq \ell(v)]\nonumber \\ 
	&\geq& \frac{1}{16} + \frac{1}{8}\left(\sum_{\substack{i \leq
	\ell^v_{max} - o_e}} x^u_{i} + \sum_{\substack{j \geq \ell^u_{min} +
o_e}} x^v_{j}\right) + \frac{z_e^2}{8}\nonumber \\
&\geq& \frac{1}{16} + \frac{1}{8}\left(\sum_{\substack{i \leq \ell^v_{max}
- o_e}} \left(\sum_{\substack{j \geq i+o_e}} y^e_{i,j}\right) + 
\sum_{\substack{j
\geq \ell^u_{min} + o_e}} \left(\sum_{\substack{i + o_e \leq j}}
y^e_{i,j}\right)\right) + \frac{z_e^2}{8}\nonumber \\
&=& \frac{1}{16} + \frac{1}{8}\left(2 \sum_{\substack{i + o_e \leq
j\\i\in S_u, j \in S_v}} y^e_{i,j}\right) + \frac{z_e^2}{8} \nonumber
\\
&\geq& \frac{1}{16} + \frac{z_e}{4} + \frac{z_e^2}{8} \nonumber \\
& = & \frac{z_e}{4}\left(1 + \frac{1}{2}\left(\frac{1}{2z_e} + 
z_e\right)\right)
\geq \frac{z_e}{4}\left(1 + \frac{1}{\sqrt{2}}\right) 
,
\end{eqnarray}
where the last inequality uses 
$t + \slfrac{1}{at} \geq \slfrac{2}{\sqrt{a}}$ for $a,t > 0$.
\qed
\end{proof}

\section{Sherali-Adams Integrality Gap for \maxk}
\label{sec:integrality-gaps}
For convenience, in the construction of the integrality gaps presented
in this section the
integral optimum and the LP objective shall count the weighted
fraction of edges satisfied.
We begin with a simple construction of an $n$-vertex digraph which is
a $\left(2 - \slfrac{2}{n}\right)$ integrality gap
for the standard LP relaxation for \maxk in Figure \ref{fig:LPmaxk},
for $2\leq k\leq n$. 
\begin{cl}\label{claim:int}
Let $D = (V,A)$ be the complete digraph on $n$ vertices, i.e. having
a directed edge for every ordered pair $(u,v)$ of distinct vertices
$u$ and $v$. Thus, $|A| = 2{n\choose 2}$. Let
$k \in \{2,\ldots, n\}$. Then,
\begin{itemize}
	\item The optimum of \maxk on $D$ is at most 
		$\frac{1}{2}\left(1 -
		\frac{1}{k}\right)\left(\frac{n}{n-1}\right)$. 
	\item There is a solution to the standard LP relaxation for
		\maxk on $D$ with value $\left(1 -
		\frac{1}{k}\right)$. 
\end{itemize}
In particular, the above implies a  $\left(2 - \slfrac{2}{n}\right)$
integrality gap
for the LP relaxation in Figure \ref{fig:LPmaxk}.
\end{cl}
\begin{proof}
The number of forward edges is simply the number of ordered pairs of vertices
$(u,v)$ with distinct labels.
By Turan's Theorem, 
the optimal integral solution is to partition the vertices into $k$
subsets whose sizes differ by at most $1$, giving each subset a distinct 
label from $\{1,\dots, k\}$. This implies that there are at most 
$\frac{n^2}{2}\left(1
- \frac{1}{k}\right)$ forward edges. Hence, the optimal integral solution
has value, 
$$\frac{1}{2}\left(1 -
\frac{1}{k}\right)\left(\frac{n}{n-1}\right).$$
On the other
hand, consider an  LP solution that assigns $x^u_i = \frac{1}{k}$ for all
$u\in V$ and $i \in [k]$, and $y^e_{i,i+1} = \frac{1}{k}$ for all
$e = (u,v) \in A$ and $i \in [k-1]$. Each edge $e$ 
contributes,
$$\displaystyle \sum_{i=1}^{k-1}y^{e}_{i,i+1} =
\left(1-\frac{1}{k}\right),$$
to the objective. \qed
\end{proof}
The above integrality gap is essentially 
retained even after near polynomial rounds of the Sherali-Adams constraints
given in Figure \ref{fig:SAmaxk}. In particular, we prove the
following that implies Theorem \ref{thm:main2}.
\begin{theorem}\label{thm:main2restated}
	For any constant $\eps > 0$, there is $\gamma > 0$ such that
	for large enough $n \in \mathbb{Z}^+$ and any $k \in \{2, \ldots, n\}$, 
	there is a weighted digraph $D^* = (V^*, A^*)$
	satisfying,
	\begin{itemize}
		\item The optimum of \maxk on $D^*$ is at most 
		$\frac{1}{2}\left(1 - \frac{1}{k}\right) + \eps$.
		\item The LP relaxation for \maxk augmented with
			$n^{\left(\slfrac{\gamma}{\log\log k}\right)}$ 
			rounds of Sherali-Adams constraints
			has objective value at least $\left(1-\eps)(1 -
			\frac{1}{k}\right).$
	\end{itemize}
\end{theorem}
The rest of this section is devoted to proving the above theorem.  Our
construction of the integrality gap uses the techniques of 
Lee~\cite{lee2014hardness} who proved a similar gap for a variant of the 
\emph{graph pricing} problem. We begin by 
showing that a sparse, random subgraph $D'$,
of the complete digraph $D$ mentioned above,
also has a low optimum solution. For this, we require the following
result on \emph{$\eps$-samples}~\cite{VC} for finite set systems that
follows from Hoeffding's bound. The reader is referred 
to Theorem 3.2 in \cite{lect} for a proof.
\begin{theorem}
\label{thm:epsilonsample}
  Let $(\mathcal{U}, \mathcal{S})$ denote a finite set system\footnote{A set system $(\mathcal{U}, \mathcal{S})$ consists of a
	  ground set $\mathcal{U}$ and a collection of its subsets
	  $\mathcal{S}\subseteq 2^{\mathcal{U}}$. It is called finite
  if $|\mathcal{U}|$ is finite.}. Suppose
  $\tilde{A}$ is a multi-set obtained by sampling
  from $\mathcal{U}$ independently and uniformly $m$ times where $m
  \geq \frac{1}{2\eps^2} \ln \frac{2 |\mathcal{S}|}{\delta}$. Then
  with probability at least $1-\delta$,
  \[\left|\frac{|\tilde{A} \cap S|}{|\tilde{A}|} -
  \frac{|S|}{|\mathcal{U}|} \right| \leq \eps, \ \ \ \forall S \in
  \mathcal{S}.\] 
  $\tilde{A}$ is referred to as an $\eps$-sample for
  $(\mathcal{U},\mathcal{S})$.
\end{theorem}
In order to construct a solution that satisfies Sherali-Adams
constraints for a large number of rounds, 
we require the instance to be \emph{locally sparse}, i.e. the
underlying undirected graph is
almost a tree on subgraphs induced by large (but bounded) 
vertex sets. We use the notion of \emph{path decomposability} 
as defined by 
Charikar et al.~\cite{charikar2009integrality} as a measure of local
sparsity.
\begin{definition}\textnormal{[Path Decomposability]} A graph $G$ is
	$l$-path decomposable if every 2-connected subgraph $H$ of $G$
	contains a path of length $l$ such that every 
	vertex of the path has degree 2 in $H$.  
\end{definition}
We proceed to show that the sparse graph $D'$ obtained as above can be
further processed so that it is locally sparse. Applying the
techniques in \cite{charikar2009integrality} and 
\cite{lee2014hardness} yields a solution with high value that
satisfies the Sherali-Adams constraints.

\subsection{Constructing a Sparse Instance}
\label{sec:constr-sparse-inst}

\begin{lemma} \label{lem:lowopt} Let $D = (V,A)$ be the complete
	digraph on $n$ vertices, let $k \in \{2, \ldots, n\}$ and
	$\eps > 0$ be a small constant. The weighted 
	digraph $D' = (V,A')$
	obtained by sampling $\Omega(\frac{n\log k}{\eps^2})$ edges
	uniformly at random satisfies $Opt(D') \leq \frac{1}{2}\left(1 -
	\frac{1}{k}\right) + \eps$ 
	with high probability, where $Opt$
	denotes the optimum of \maxk.
\end{lemma}
\begin{proof} Let $V = [n]$ and $A \subseteq [n] \times [n]$ denote the
	vertices and edges of the digraph $D$. Let $\rho: [n]
	\rightarrow [k]$ denote some labeling of the vertices.
	Let $S_\rho := \{(i,j)\,\mid\,\rho(i) < \rho(j), (i,j) \in A\}$
	denote the subset of edges that are satisfied by
	$\rho$, and $\mathcal{S} := \{S_\rho\,\mid\,\forall \text{
	labelings } \rho\}$ denote the collection of such subsets
	induced by all feasible labelings. Since the number of
	distinct labelings is $k^n$, we have that $|\mathcal{S}| \leq
	k^n$. 

	We now construct an $(\eps/2)$-sample for the set system $(A,
\mathcal{S})$ by randomly sampling edges (with replacement) as per
Theorem \ref{thm:epsilonsample}. Let $\tilde{A}$ denote the bag of
randomly chosen $\left(\frac{2}{\eps^2} \ln \frac{2k^n}{\delta}\right)$ 
edges. Substituting $\delta = \frac{1}{n}$, we get that $|\tilde{A}| =
\Omega(\frac{n\log k}{\eps^2})$ and with probability at least $1 -
\frac{1}{n}$ we have,
\begin{equation} \label{eq:lowobjective} \left|\frac{|S_\rho \cap
\tilde{A}|}{|\tilde{A}|} - \frac{|S_\rho|}{|A|}\right| \leq \eps/2, \ \
\ \forall \rho. 
\end{equation} 
In order to avoid multi-edges in the construction, we define the
weight of an edge $w(u,v)$ to be the number of times that edge is
sampled in $\tilde{A}$ and let $A'$ denote the set of thus weighted 
edges obtained
from $\tilde{A}$. Equation
\eqref{eq:lowobjective} along with Claim \ref{claim:int}
guarantees that the optimum integral
solution of the weighted graph $D'$ induced by the edges $A'$ is
bounded by $Opt(D') \leq Opt(D) + \eps \leq \frac{1}{2}\left(1 -
\frac{1}{k}\right)\left(\frac{n}{n-1}\right) + \eps/2 \leq 
\frac{1}{2}\left(1 - \frac{1}{k}\right) + \eps$ as desired.
\qed
\end{proof}
Given a digraph, let its corresponding undirected multigraph be
obtained by replacing every directed edge by the corresponding
undirected one. Note that if the digraph contains both $(u,v)$ and
$(v,u)$ edge for some pair of vertices, then the undirected multigraph 
contains two parallel edges between $u$ and $v$. 

We now show that $D'$ 
obtained in Lemma \ref{lem:lowopt} 
can be slightly modified 
so that its corresponding underlying multigraph is almost
regular, has high girth, and is locally sparse i.e. all small enough
subgraphs are $l$-path decomposable for an appropriate choice of
parameters. 

\begin{lemma}
	\label{lem:girth} Let $k = \{2,\dots, n\}$ and $\eps > 0$ be
	a small enough constant. Given the complete digraph 
	$D = (V,A)$ on $n$ vertices, let $D' = (V,A')$ 
	be obtained by sampling (with replacement) 
  $\Theta(\frac{n\log k}{\eps^2})$ edges uniformly at random. 
  Then, with high probability there exists 
  a subgraph $D'' = (V,A'')$ of $D'$ obtained 
  by removing at most $\eps |A'|$ edges, such that the undirected 
  multigraph $G''$ underlying $D''$ satisfies the following
  properties: 
  \begin{enumerate}
  \item \emph{Bounded Degree:} The maximum degree of any vertex is 
	  at most $2 \Delta$ and $G''$ has $\Omega(\Delta n)$ edges,
	  where $\Delta = \Theta(\frac{\log k}{\eps^2})$.
  \item \emph{High Girth:} $G''$ has girth at least $l = O(\frac{\log
	  n}{\log \Delta})$.
  \end{enumerate}
\end{lemma}

\begin{proof}
Since $D'$ is obtained by sampling $\Theta(\frac{n\log k}{\eps^2})$ edges
uniformly, the probability that any given edge is selected is $p =
\Theta(\frac{\log k}{\eps^2 n})$. In addition, these events are negatively
correlated. Therefore given any set of edges $S$, the probability that all
the edges in $S$ are sampled is upper bounded by $p^{|S|}$.

\emph{Bounded Degree:} As the maximum degree of any vertex in $D$ is
at most $2n$, the expected degree of any vertex $v \in V$ in $D'$ is
at most $\Delta = 2pn = \Theta(\frac{\log k}{\eps^2})$. Call a vertex $v
\in V$ \emph{bad} if it has degree more than $2\Delta$ in $D'$, and
call and edge $(u,v) \in A'$ bad if either $u$ or $v$ is bad. Now, for
any edge $(u,v)$, the probability that $(u,v)$ is bad given that
$(u,v) \in A'$ is at most $2e^{-\frac{\Delta}{3}}$ by Chernoff
bound. Hence, the expected number of bad edges is at most
$2e^{-\frac{\Delta}{3}}|A'|$. Finally, by Markov's inequality, with
probability at least $\frac{1}{2}$, the number of bad edges is at most
$4e^{-\frac{\Delta}{3}}|A'|$. Deleting all bad edges guarantees that
the maximum degree of $D'$ is at most $2\Delta$ and with probability
at least half, we only delete $4e^{-\frac{\Delta}{3}}|A'|$ edges which
is much smaller than $\eps |A'|$ since $\Delta =
\Theta(\frac{\log k}{\eps^2})$.

\emph{Girth Control:} Let $G'$ denote the undirected multigraph underlying $D'$. Since the degree of any vertex in $D$ is at most $2n$, we have
\begin{align*}
  \mathbb{E}[\text{Number of cycles in $G'$ of length $i$}] &\leq n(2n)^{i-1}p^i \leq (C \Delta)^i\\
\intertext{for some constant $C$. For $i = O(\frac{\log n}{\log \Delta})$, we get}
  \mathbb{E}[\text{Number of cycles in $G'$ of length $i$}] &\leq n^{0.5}
\end{align*}

Summing up over all $i$ in $2 \ldots l = O(\frac{\log n}{\log \Delta})$, we get that the expected number of cycles of length up to $l$ is at most $O(n^{0.6})$ and hence it is less than $O(n^{0.7})$ with high probability. We can then remove one edge from each such cycle (i.e. $o(n)$ edges) to ensure that the graph $G''$ so obtained has girth at least $l$. In particular, note that the corresponding digraph $D''$ has no 2-cycles.
\qed  
\end{proof}

\subsubsection{Ensuring Local Sparsity}
\label{sec:ensur-local-spars}
Using Lemma \ref{lem:girth} 
we ensure that the subgraph of $G''$ induced by any subset of 
$n^\delta$ vertices is $l$-path decomposable for some constant $\delta > 0$.
The following lemma shows that $2$-connected subgraphs of $G''$ of
polynomially bounded size are sparse.
\begin{lemma} \label{lem:local-sparsity} The undirected multigraph
	$G''$ underlying the digraph $D''$ satisfies the following p,
	i.e., there exists $\delta > 0$ such that every 2-connected
	subgraph $\tilde{G}$ of $G''$ containing $t' \leq n^\delta$
vertices has only $(1+\eta)t'$ edges where $\eta = \frac{1}{3l}$.
\end{lemma}

\begin{proof} Let $G$ denote the undirected multigraph underlying the
	graph $D$ that was used as a starting point for Lemma
	\ref{lem:girth}. The proof proceeds by counting the number of
	possible ``dense'' subgraphs of $G$ and showing that the
	probability that any of them exist in $G''$ after the previous
	sparsification steps is bounded by $o(1)$.  We consider two
	cases based on the value of $t'$. 

\emph{Case 1: $4 \leq t' \leq \frac{1}{\eta}$}. We first bound the
total number of 2-connected subgraphs of $G$ with $t'$ vertices and
$t'+1$ edges. It is easy to verify that the only possible degree
sequences for such subgraphs are $(4,2,2,\ldots)$ or
$(3,3,2,2,\ldots)$. Suppose it is $(4,2,2,\ldots)$ and let $v$ be the
vertex with degree 4. Now, there must be a sequence of $t'+2$ vertices
$(v, \ldots, v, \ldots, v)$ that represents an Eulerian tour. But the
number of such sequences is upper bounded by $nt'(n)^{t'-1} =
t'n^{t'}$ ($n$ for guessing $v$, $t'$ for guessing the position of $v$
in the middle, and $n^{t'-1}$ to guess the other $t'-1$ vertices).
Now assume that the degree sequence is $(3,3,2,2,\ldots)$ and $u,v$ be
the vertices with degree 3. Now, there must a sequence of $t+2$
vertices $(u, \ldots, v, \ldots, u, \ldots, v)$ that represents an
Eulerian path from $u$ to $v$. By a similar argument, the number of
such sequences is bounded by ${t'}^2n^{t'}$.  Hence, we are guaranteed
that the total number of 2-connected subgraphs of $G$ with $t'$
vertices and $t'+1$ edges is at most $2{t'}^2n^{t'}$.

Therefore, the probability that there exists a subgraph of $G''$ with
$t'$ vertices and $t'+1$ edges for $4 \leq t' \leq \frac{1}{\eta} =
3l$ is at most \begin{equation} \label{eq:sparse-small}
	\sum_{t'=4}^{3l}  2{t'}^2n^{t'}p^{t'+1} = \sum_{t'=4}^{3l}
	2{t'}^2n^{t'}\left(\frac{\Delta}{2n}\right)^{t'+1} \leq \frac{C}{n}l^3
	\left(\frac{\Delta}{2}\right)^{3l+1}, \end{equation} 	
where $C$ is an appropriate constant. 
For $l = O(\frac{\log n}{\log \Delta})$, we have that
$(\frac{C}{n})l^3(\frac{\Delta}{2})^{3l+1} \leq n^{-0.1} = o(1)$.

\emph{Case 2: $n^\delta \geq t' > \frac{1}{\eta} = 3l$.} In this case,
we count the number of subgraphs of $G$ with $t'$ vertices and
$(1+\eta)t'$ edges. As shown by Lee~\cite{lee2014hardness}, the number
of such subgraphs is bounded by $C
\alpha^{t'}n^{t'}(\frac{et'}{2\eta})^{2\eta t'}$ for some constants
$C$ and $\alpha$.

Therefore, the probability that such a subgraph exists in $G''$ is at most
\begin{eqnarray}
  \label{eq:sparse-large}
  C \alpha^{t'}n^{t'}\left(\frac{et'}{2\eta}\right)^{2\eta t'} 
  p^{(1+\eta)t'} & = & C
  \alpha^{t'}n^{t'}\left(\frac{et'}{2\eta}\right)^{2\eta t'} 
  \left(\frac{\Delta}{2n}\right)^{(1+\eta)t'} \nonumber \\
  & \leq & (C_1\Delta^2)^{t'} \left(C_2 \frac{l^2{t'}^2}{n}\right)^{t'/3l},
\end{eqnarray}
where $C_1$ and $C_2$ are appropriate constants. We choose $l =
O(\frac{\log n}{\log \Delta})$ and $\delta \in (0,0.1)$, such that
above quantity is less than $n^{-0.1}$. Summing up over all $t' =
3l,\ldots,n^\delta$, we still have that probability such a subgraph
exists is bounded by $o(1)$. \qed \end{proof}

Finally, we need the following lemma proved by Arora 
et al.~\cite{arora2002proving} regarding the existence of long paths in
sparse, 2-connected graphs.  
\begin{lemma}[Arora et al.~\cite{arora2002proving}] 
	\label{lem:arora} Let
	$l\geq 1$ be an integer and $0 < \eta < \frac{1}{3l-1}$, and
	let $H$ be a 2-connected graph with $t$ vertices and at most
	$(1+\eta)t$ edges and $H$ is not a cycle. Then $H$ contains a
	path of length at least $l+1$ whose internal vertices have
	degree 2 in $H$.  
\end{lemma}

\begin{corollary} \label{cor:path-decomposable} Every subgraph
$\tilde{G}$ of $G''$ that is induced on at most $t' \leq n^\delta$
vertices is $(l-1)$-path decomposable.  
\end{corollary} 
\begin{proof}
	Consider any 2-connected subgraph $H$ of $\tilde{G}$. If $H$
	is not a cycle, then Lemma \ref{lem:local-sparsity} and Lemma
	\ref{lem:arora} together guarantee that $H$ contains a path of
	length at least $l+1$ such that all internal vertices have
	degree $2$ in $H$, which gives us a path of length $(l-1)$
	with all vertices of degree $2$ in $H$. 
	On the other hand, if $H$ is a cycle, then Lemma
	\ref{lem:girth} guarantees that $H$ has at least $l+1$
	vertices and hence again the required path exists.  \qed
\end{proof}
For convenience we replace $(l-1)$ in Corollary
\ref{cor:path-decomposable} with $l$, and since $l = \Theta(\frac{\log
n}{\log \Delta})$, this does not change any parameter noticeably.

\subsubsection{Final Instance}

\begin{theorem} 
	\label{thm:instance} Given $k \in \{2,\dots, n\}$
	and constants $\eps, \mu > 0$, there exists a constant $\gamma
	> 0$, and parameters $\Delta = \Theta(\frac{\log k}{\eps^2})$,
	and $l = \Theta(\frac{\log n}{\log \Delta})$ such that there
	is an instance $\hat{D}$ (with underlying undirected 
	graph $\hat{G}$) of \maxk with the following properties 
	\begin{itemize} 
		\item Low Integral Optimum: $Opt(\hat{D}) \leq
			\frac{1}{2}(1 - \frac{1}{k}) + \eps$. 
		\item Almost Regularity: Maximum Degree of $\hat{G}
			\leq 2 \Delta$, and $\hat{G}$ has 
			$\Omega(\Delta n)$ edges.  
		\item Local Sparsity: For $t < n^{\gamma/\log \Delta}$, every
			induced subgraph of $G$ on $(2 \Delta)^l t$ 
			vertices is $l$-path decomposable.  
		\item Large Noise: For $t < n^{\gamma/\log \Delta}$, 
			$(1-\mu)^{l/10} 
			\leq \frac{\mu}{5t}$. 
	\end{itemize}
	Note that $n^{\gamma/\log \Delta} =
	n^{\Omega\left(\slfrac{1}{\log\log k}\right)}$. 
\end{theorem}
\begin{proof}
	Let $D''$ be the digraph obtained from Lemma \ref{lem:girth}.  
	Lemmas \ref{lem:lowopt} and \ref{lem:girth} imply that  
the digraph $D''$ so obtained
(i) has low integral optimum, (ii) is almost regular, and (iii) has
girth $\geq l$. 

The large noise condition is satisfied by $l \geq \left(\slfrac{C
\gamma \log n}{\log \Delta}\right)$ for an appropriate constant $C$. 

Corollary \ref{cor:path-decomposable} guarantees that the local
sparsity condition is satisfied if $(2 \Delta)^lt \leq n^\delta$, i.e.
$l \leq C'(\delta - \gamma) \log n$ for another constant $C'$. Hence,
by selecting a small enough constant $\gamma$ and an appropriate $l =
\Theta(\frac{\log n}{\log \Delta})$, 
the instance $D''$ obtained in Lemma \ref{lem:girth}
satisfies all the required properties. \qed 
\end{proof}

\subsection{Constructing Local Distributions}
\label{sec:constr-local-distr}

Let $D = (V,A)$ be the instance of \maxk constructed in Theorem
\ref{thm:instance} and let $G=(V,E)$ be the underlying undirected
graph. We now show that there exists a solution to the LP after $t =
n^{\gamma/\log \Delta}$ 
rounds of the Sherali-Adams hierarchy whose objective is at
least $(1 - \eps)(1 - \frac{1}{k})$. Our proof for the existence of
such a solution essentially follows the approach of 
Lee~\cite{lee2014hardness}.  
Given a set of $t \leq n^{\gamma/\log \Delta}$ vertices $S
= \{v_1, v_2, \ldots, v_t\}$, our goal is to give a distribution on
events $\{\ell(v_1) = x_1, \ell(v_2) = x_2, \ldots, \ell(v_t) =
x_t\}_{x_1,x_2,\ldots,x_t \in [k]}$.

Let $d(u,v)$ be the shortest distance between $u$ and $v$ in the
(undirected) graph $G$. Let $V' \subset V$ be the set of vertices that
are at most $l$ distance away from $S$ and let $G'$ be the subgraph
induced by $V'$ on $G$. Since the maximum degree of vertices is
bounded by $2\Delta$, we have $|V'| \leq (2\Delta)^lt$ and hence $G'$
is $l$-path decomposable by Theorem \ref{thm:instance}.

The first step of the construction relies on the following theorem by
Charikar et al.~\cite{charikar2010local} that shows that if a graph
$G'$ is $l$-path decomposable, then there exists a distribution on
partitions of $V$ such that close vertices are likely to remain in the
same partition while distant vertices are likely to be separated. 

\begin{theorem}[Charikar et al.~\cite{charikar2010local}]
\label{thm:cmmmulticut}
Suppose $G' = (V',E')$ is an $l$-path decomposable graph. Let
$d(\cdot,\cdot)$ be the shortest path distance on $G$ , and $L =
\lfloor l/9 \rfloor$; $\mu \in [1/L, 1]$. Then there exists a
probabilistic distribution of multicuts of $G'$ (or in other words
random partition of $G'$ into pieces) such that the following
properties hold. For every two vertices $u$ and $v$, 
\begin{enumerate}
	\item If $d(u,v) \leq L$, then the probability that $u$ and
		$v$ are separated by the multicut (i.e. lie in
		different parts) equals $1 - (1 - \mu)^{d(u,v)}$;
		moreover, if $u$ and $v$ lie in the same part, then
		the unique shortest path between $u$ and $v$ also lies
		in that part.  
	\item If $d(u,v) > L$, then the probability that $u$ and 
		$v$ are separated by the multicut is at least 
		$1 - (1-\mu)^L$.  
	\item Every piece of the multicut partition is a tree.  
\end{enumerate} 
\end{theorem}
Based on this random partitioning, we define a distribution on the
vertices in $S$ (actually in $V'$). As each piece of the above
partition is a tree, given some vertex $u$ with an arbitrary label
$i$, we can extend it to a labeling $\ell$ for 
every other vertex in that piece such that every
directed edge $(x,y)$ in the piece satisfies $\ell(y) - \ell(x) = 1\ 
(\textnormal{mod }k)$.

For vertices $u$ and $v$ with $d(u,v) \leq L$, we say that label $i$
for $u$ and $i'$ for $v$ \emph{match} if the labeling $\ell(u) = i,
\ell(v)
= i'$ can be extended so that for every directed edge $(x,y)$ on the
unique shortest path between $u$ and $v$, $\ell(y) - \ell(x) = 1\ 
(\textnormal{mod } k)$. Note
that there are exactly $k$ such matching pairs for every $u$ and $v$.
We can now use Theorem \ref{thm:cmmmulticut} to obtain a random
labeling as follows.
\begin{corollary} \label{cor:tree} Suppose $G' = (V',E')$ is an
	$l$-path decomposable graph. Let $L = \lfloor l/9 \rfloor; \mu
	\in [1/L,1]$. Then there exists a random labeling $r: V'
	\rightarrow [k]$ such that 
	\begin{enumerate} 
		\item If $d = d(u,v) \leq L$, then 
			\\ $\Pr[ r(u) = i, r(v) = i'] = 
			\begin{cases} \frac{(1-\mu)^d}{k} + \frac{1
					- (1-\mu)^d}{k^2} &\text{if
					$i$ and $i'$ match}\\
					\frac{1-(1-\mu)^d}{k^2}
					&\text{otherwise} 
			\end{cases}$
		\item If $d > L$, then
			\\ $\frac{1-(1-\mu)^L}{k^2} \leq
			\Pr[ r(u) = i, r(v) = i'] \leq  
			\frac{1 - (1-\mu)^L}{k^2} +
			\frac{(1-\mu)^L}{k}$ for any 
			$i, i' \in [k]$ 
	\end{enumerate}
\end{corollary}
\begin{proof}
We first sample from the distribution of multicuts given by Theorem
\ref{thm:cmmmulticut}. For every piece obtained, we pick an arbitrary
vertex $u$ and assign $r(u)$ to be a uniformly random label from
$[k]$. Now, since each piece is a tree, we can propagate this label
along the tree so that for every directed edge $(v,w)$ we have $r(w) -
r(v) = 1\ (\textnormal{mod }k)$. Note that the final distribution obtained
does not depend on the choice of the initial vertex $u$.

Consider any two vertices $u$ and $v$. If $d(u,v) \leq L$, then if $u$ and $v$ are in the same piece, then the path connecting $u$ and $v$ in the piece is the shortest path. If $i$ and $i'$ are matching labels, then
\begin{eqnarray}
\Pr[ r(u) = i, r(v) = i'] & = & \Pr[\text{$u,v$ in the same
piece}]\cdot \left(\frac{1}{k}\right) \nonumber \\
& & + \Pr[\text{$u,v$ are separated}]\cdot\left(\frac{1}{k^2}\right).
\nonumber 
\end{eqnarray}
On the other hand, if $i$ and $i'$ are not matching,
\begin{eqnarray}
\Pr[ r(u) = i, r(v) = i'] & = & \Pr[\text{$u,v$ in the same
piece}]\cdot 0 \nonumber \\ 
& & + \Pr[\text{$u,v$ are
separated}]\cdot\left(\frac{1}{k^2}\right).\nonumber
\end{eqnarray}
Similarly, if $d(u,v) > L$, then $\Pr[ r(u) = i, r(v) = i']$ is lower
bounded by $\slfrac{\Pr[\text{$u,v$ are separated}]}{k^2}$ and upper
bounded by $\slfrac{\Pr[\text{$u,v$ in the same piece}]}{k} +
\slfrac{Pr[\text{$u,v$ are separated}]}{k^2}$.  Substituting the
separation probabilities in Theorem \ref{thm:cmmmulticut} proves the
desired result. \qed \end{proof}

The above random labeling defines a distribution $\nu_S$ over labels of pairs of vertices as follows.
\begin{definition} Let $S = \{v_1, v_2, \ldots, v_t\}$ be a fixed set
	of vertices. For any two vertices $u,v \in S$ and $i,i' \in
	[k]$, let $\nu_S(u(i), v(i')) = \Pr[ x(u) = i, x(v)=i']$ in the local distribution on $S$ defined by $r$ in Corollary \ref{cor:tree}.
\end{definition}

We now define another distribution $\rho$ over labels for pairs of
vertices that is independent of the choice of the set $S$ as follows.
\begin{definition} For any vertices $u \neq v$ and $i,i' \in [k]$, let
	$\rho(u(i),v(i')) = \Pr[ x(u) = i, x(v)=i']$ if $d(u,v) \leq
	L$, and $\frac{1}{k^2}$ otherwise. Also define $\rho(u(i),
	u(i)) = \frac{1}{k}$ and $\rho(u(i), u(i')) = 0$ for $i \neq
	i'$. Since the shortest path between $u$ and $v$ is unique
	when $d(u,v) \leq L$, $\rho$ is uniquely defined by $D$ and
	$G$ and is independent of the choice of set $S$.
\end{definition}
Lee~\cite{lee2014hardness} shows that it is possible to use the $\rho$
and $\nu_S$ distributions defined above to produce consistent
distributions over events of the form $\{\ell(v_1)=x_1, \ldots,
\ell(v_t)=x_t\}_{x_1, \ldots, x_t \in [k]}$. Further, these distributions
need to be consistent, i.e., the marginal distribution on $S \cap S'$
does not depend on the choice of its superset ($S$ or $S'$) that is
used to obtained the larger local distribution. The key idea here as
shown by Charikar et al.~\cite{charikar2009integrality} is to embed
$\rho$ into Euclidean space 
with a small error to obtain $tk$ vectors $\{v(i)\}_{v
\in S, i \in [k]}$ such that $u(i) \cdot v(i') \approx \rho(u(i),
v(i'))$. This uses the large noise property in Theorem
\ref{thm:instance}. The following lemma appears as Lemma 5.7 in
\cite{lee2014hardness}.
\begin{lemma}[Lee~\cite{lee2014hardness}]
There exist $tk$ vectors $\{v(i)\}_{v \in S, 
i \in [k]}$ such that $||v(i)||_2^2 = \mu + \frac{1}{T+1}$ and $u(i) \cdot v(i') = \frac{\mu}{2} + \rho(u(i), v(i'))$.
\end{lemma}
Given such $tk$ vectors, one can use a geometric rounding scheme to
define the consistent local distributions. Note that the local
distribution is completely defined by the pairwise inner products of
the vectors which, for any two vectors, is independent of the
subset $S$. Lee~\cite{lee2014hardness} 
shows that the
following simple rounding scheme suffices to obtain a good
distribution: choose a random Gaussian vector $g$, and for each vertex
$v$, let $\ell(v) = \argmax_i (v(i) \cdot g)$.  
\begin{lemma}[Lee \cite{lee2014hardness} \footnote{The lemma follows
	from the proof of Lemma 5.8 of Lee \cite{lee2014hardness} by
	substituting $l_A(u,v) = 1$.}] \label{lem:lee-obj} There
	exists a $\mu > 0$ depending on $k$ and $\eps$ such that, in
	the above rounding scheme, for any edge $(u,v)$ and any label
	$i \in [k]$ the probability that $\ell(u)=i$ and $\ell(v)=i+1\ 
	(\text{mod
	$k$})$ is at least $\frac{1-12\eps}{k}$.  
\end{lemma}
	Consider the solution to $n^{\gamma/\log \Delta}$ 
	rounds of the Sherali-Adams hierarchy obtained by the 
	above rounding process. 
	For any edge $(u,v) \in A$, its contribution 
	to the objective is 
	\begin{eqnarray}
	\sum_{1\leq i<i'\leq k} 
	\Pr[ \ell(u) = i, \ell(v) = i'] & \geq & 
	\sum_{i \in [k-1]} \Pr[ \ell(u)=i, \ell(v) = i+1] \nonumber \\ 
	& \geq & \sum_{i \in [k-1]}\frac{1-12\eps}{k}\nonumber 
	\end{eqnarray}
The last inequality follows due to  Lemma \ref{lem:lee-obj}. Thus we have 
a fractional solution with  value at least $(1-12\eps)(1 - \frac{1}{k})$.
This, along with the low optimum of the instance from Theorem
\ref{thm:instance} completes the proof of Theorem
\ref{thm:main2restated}.

\section{The \mink Problem}
\label{sec:mink}

Recall that the \mink problem is to remove the minimum weight subset 
of edges from a given DAG
so that the remaining digraph
does not contain any path of length $k$. 

\subsection{Combinatorial $k$-Approximation}
\label{sec:greedy-k-appr}

In the unweighted case (i.e. all edges have unit weight), 
the following simple
scheme is a $k$-approximation algorithm. As long as the DAG
contains a directed path $P$  of length $k$, delete \emph{all} edges
of that path. It is easy to see that the above scheme guarantees a
$k$-approximation as the optimal solution must delete at least one
edge from the path $P$ while the algorithm deletes exactly $k$ edges.

The following slightly modified scheme that uses the 
local ratio technique yields a $k$-approximation for weighted DAGs.

\vspace{2mm}
{\bf Algorithm {\sf LocalRatio}}:
\vspace{-2mm}
\begin{enumerate}
\item $S \leftarrow \{e \in E \ | \ w(e) = 0\}$
\item While $(V, E \setminus S)$ contains a path $P$ of length $k$
  \begin{enumerate}
  \item $w_{min} \leftarrow \min_{e \in P} (w(e))$
  \item $w(e) = w(e) - w_{min}, \forall e \in P$
  \item $S = \leftarrow \{e \in E \ | \ w(e) = 0\}$
  \end{enumerate}
\end{enumerate}

\begin{theorem}
	{\sf LocalRatio} 
	is a polynomial time $k$-approximation to the 
	\mink problem on weighted DAGs. 
\end{theorem}
\begin{proof}
	We note that the {\sf LocalRatio} 
	terminates in at most $|E|$ iterations as
the weight of at least one edge reduces to 0 in each iteration. Also,
since one can check if there exists a path of length $k$ in DAG via a
dynamic programming, it follows that {\sf LocalRatio} runs in polynomial
time.

Let $\mc{O}\subseteq E$ be an optimal solution and $\mc{S}\subseteq E$ be the
solution returned by {\sf LocalRatio}.
Note that an edge is in $\mc{S}$ if its
weight is reduced to $0$ in some iteration of the algorithm. 
Thus, the weight of $\mc{S}$ is upper
bounded by the total reduction in the weight of the edges. 
At each iteration, for a path $P$ of
length $k$, the reduction is at most
$k$ times the minimum weight edge (according to the current weights)
on in $P$. Since there is at least one edge $e$ in $P$ which is in
$\mc{O}$, we charge this reduction to the weight of $e$. Then
the weight of $e$ decreases by at least $1/k$ factor of what is charged to
it, and it cannot decrease beyond $0$. Thus, the weight of $\mc{S}$ is
at most the $k$ times the weight of $\mc{O}$. \qed
\end{proof}

\subsection{$k$-Approximation via LP Rounding}
\label{sec:ded-lp}

The natural LP relaxation for \mink on an $n$-vertex DAG $D = (V,E)$ 
is given in Figure \ref{fig:LPmink}.
\begin{figure}[h]
\setlength{\fboxsep}{10pt}
\begin{center}
\fbox{
\begin{minipage}[c]{4.3in}
\begin{equation}
	\min
	\sum_{e \in E}w(e)x_e
\end{equation}
subject to,
\begin{align}
	& \forall \text{ paths $P$ of length $k$} , && \displaystyle
	\sum_{e \in P} x_e	\geq 1 .\label{eqn:LP-min-sum}\\
        &\forall e \in E, && x_e \geq 0.	\label{eqn:LP-min-xpositive}
\end{align}
\caption{LP Relaxation for instance $\mc{I}$ of \mink.}
\label{fig:LPmink}
\end{minipage}
}\end{center}
\end{figure}
This relaxation has $n^{O(k)}$ constraints. 
However, when the input graph is a DAG, it admits the following 
polynomial time 
separation oracle for any $k \in \{2,\dots, n-1\}$.

\subsubsection{Separation Oracle and Rounding.}
\label{sec:separation}
For each vertex $v \in V$ and integer $t \in [n]$, 
define ${a^v_t = \min_{P} (\sum_{e \in P} x_e)}$ 
where $P$ is a path of length $t$ that ends at vertex $v$. 
Once we compute all these $a^v_t$ values, then a constraint is 
violated if and only if there is a vertex $v$ such that $a^v_k < 1$.

On a DAG the $\{a^v_t\,\mid\, v \in V, t \in [n]\}$ can be computed 
by dynamic programming. First assume that the vertices are arranged in 
a topological order. For any vertex $v$ with no predecessors, set 
$a^v_t = 0, \forall t$. Otherwise, we have the following recurrence,
\[ a^v_t = \min_{u \in \text{predecessors($v$)}} (x_{(u,v)} +
a^u_{t-1}).\]
It is easy to see that the above recurrence leads to a dynamic program
on a DAG.
Once we obtain an optimal solution to the LP relaxation, a simple threshold based rounding using  a threshold of $\slfrac{1}{k}$ yields a $k$-approximation.
\begin{theorem}
	The standard LP relaxation for \mink on $n$-vertex DAGs can be
	solved in polynomial time for $k=\{2,\dots, n-1\}$ and yields
	a $k$-approximation. 
\end{theorem}

\subsection{Hardness of Approximation}
\label{sec:mink-hardness}

For fixed integer $k \geq 2$ and arbitrarily small constant $\eps >
0$,
Svensson~\cite{svensson2012hardness} showed factor $(k-\eps)$
UGC-hardness of
the \emph{vertex deletion} version of \mink, which requires deleting
the minimum number of vertices from a given DAG to remove all paths
with $k$ vertices. 
In particular, 
\cite{svensson2012hardness} proves
the following structural hardness result.
\begin{theorem}[Svensson~\cite{svensson2012hardness}]
	\label{thm:svensson} For any fixed integer $t \geq 2$ and arbitrary
	constant $\eps > 0$, assuming the UGC the following is NP-hard: 
	Given a DAG
	$D(V,E)$, distinguish between the following cases:
	\begin{itemize} 
		\item \textnormal{(Completeness):} 
			There exist $t$ disjoint subsets
			$V_1, \ldots, V_t \subset V$ satisfying $|V_i|
			\geq \frac{1-\eps}{t}|V|$ and such that a
			subgraph induced by any $t-1$ of these subsets
			has no directed path of $t$ vertices.  
		\item \textnormal{(Soundness):} 
			Every induced subgraph on $\eps|V|$
			vertices has a path with $|V|^{1-\eps}$ vertices.  
	\end{itemize}
\end{theorem}
The following theorem provides a simple gadget reduction from the
above theorem to a hardness for \mink on DAGs.
\begin{theorem}
  \label{thm:ded-hardness}
  Assuming the UGC, for any constant $k \geq 4$
  and $\eps > 0$, 
  the \mink problem on weighted DAGs is NP-hard to approximate with a
  factor better than $\left(\lfloor\slfrac{k}{2}\rfloor - \eps\right)$.
\end{theorem}
\begin{proof}
Fix $t = \lfloor\slfrac{k}{2}\rfloor$. Let $D = (V, E)$ be a hard
instance from Theorem \ref{thm:svensson} for the parameter $t$ and
small enough $\eps > 0$. 
The following simple reduction yields a weighted DAG $H
= (V_H, E_H)$ as an instance of \mink. 
Assign $w(e) = 2|V|$ to every edge $e \in E$. Split
every vertex $v \in V$ into $v_{in}$ and $v_{out}$ and add a directed edge
$(v_{in}, v_{out})$ of weight $1$. Also every edge entering $v$ now
enters $v_{in}$ while edges leaving $v$ now leave $v_{out}$. It is
easy to see that removing all edges of weight $1$ from $H$ eliminates
all paths with $2$ edges, implying that the optimum solution has
weight at most $|V|$. Thus, we may assume that the optimum solution
does not delete any edge of weight $2|V|$. 

We now show that Theorem \ref{thm:svensson} implies that it is
$UG$-hard to distinguish whether: (Completeness) $H$ has a
solution of cost $\leq (\frac{1}{t} + \eps)|V|$, or (Soundness) $H$ 
has no solution of cost $(1 - \eps)|V|$. This immediately implies the
desired $(t -\eps') = \left(\lfloor\slfrac{k}{2}\rfloor - \eps'\right)$
UGC-hardness for \mink.
\begin{itemize} 
	\item (Completeness) There exists a subset $S \subseteq V$ of
			size at most $(\frac{1}{t}+\eps)|V|$, such
			that removing $S$ eliminates all paths in $D$
			of $t$ vertices. Let $S'$ denote
			the set of edges in $H$ corresponding to the
			vertices in $S$. It is easy to observe that
		 	$H(V_H, E_H \setminus S')$
			has no paths of length (number of edges) $2t$.
			Thus, $S'$ is a feasible solution to the \mink
			problem of cost $(\frac{1}{t} + \eps)|V|$.
		
	\item (Soundness) Assume for the sake of contradiction, that
		we have an optimal solution $S' \subseteq E_H$ of cost at
		most $(1 - \eps)|V|$. Since $S'$ is an optimal
		solution it only has edges of weight $1$, each
		of which correspond to a vertex in $V$.
		Let $S$ denote this set of vertices in $V$. By
		construction, since $H(V_H, E_H \setminus
		S')$ has no paths with $k$ edges, $D[V \setminus S]$
		has no induced paths with $
		\lfloor\slfrac{k}{2}\rfloor + 1 = t + 1$ vertices. 
		Further,
		since $|S'| = |S| \leq (1 - \eps)|V|$, we have $|V
		\setminus S| \geq \eps |V|$. Thus, we have a set of size
		$\eps |V|$ that has no induced paths of length $t+1$.
		This is a contradiction since every induced subgraph
		of $\eps |V|$ vertices has a path of length
		$|V|^{1-\eps} \geq t+1$.  
\end{itemize} 
\qed 
\end{proof}

\bibliographystyle{plain}
\bibliography{bibfile}

\end{document}